\documentclass[onecolumn]{IEEEtran}
\usepackage[T1]{fontenc}
\usepackage[utf8]{inputenc}

\usepackage{amsmath,amssymb,amsfonts,amsthm,mathtools}
\usepackage{bm}
\usepackage{bbm}
\usepackage{mathrsfs}
\usepackage{dsfont}
\usepackage{cancel}

\usepackage{graphicx}
\usepackage{booktabs}
\usepackage{cite}
\usepackage[hidelinks]{hyperref}
\usepackage[nameinlink,noabbrev,capitalise]{cleveref}
\usepackage{booktabs}
\usepackage{subcaption}
\usepackage{siunitx}
\newtheorem{theorem}{Theorem}
\newtheorem{lemma}{Lemma}

\newtheorem{corollary}{Corollary}
\theoremstyle{definition}
\newtheorem{definition}{Definition}
\newtheorem{remark}{Remark}

\title{Information Rate Decomposition for Noisy Nanopore Channels with Geometric Duplication}

\author{Brendon McBain and Emanuele Viterbo
\thanks{The authors are with the Department of Electrical and Computer Systems Engineering, Monash University, Melbourne, VIC 3800, Australia (e-mail: \texttt{brendon.mcbain@monash.edu}; \texttt{emanuele.viterbo@monash.edu}). A preliminary version of part of this work appeared in \cite{McBain2026}.}}

\begin{document}

\maketitle

\begin{abstract}
This paper studies information rates of noisy duplication channels with memory, motivated by nanopore DNA sequencing. In nanopore sequencing, the measured signal is affected by both inter-symbol interference (ISI), caused by multiple DNA bases residing in the pore, and random sample duplications, where variable translocation speed causes each base to generate a random number of samples. These two effects make direct theoretical analysis difficult. To address this, we derive a new decomposition of the information rate into two interpretable terms: one {capturing the channel memory through an auxiliary ISI channel}, and another {capturing} the uncertainty in the segment boundaries caused by random duplications. This decomposition separates the dominant channel distortions and replaces the direct analysis of the full channel with two more readily {tractable} components. We then study the second term through a soft alignment functional closely related to Soft-DTW, which {yields a strong asymptotic equipartition property result and an alternative proof of the Markov-constrained coding theorem}. Finally, we develop a lower bound on the information rate
{via a first-order, genie-aided upper bound on
the segmentation penalty} that depends on the distribution
of jump distances between adjacent nanopore levels. This bound gives a simple geometric explanation of channel synchronisability and provides a tractable framework for computing achievable rates {under the assumed channel model
using ONT pore models}.
\end{abstract}

\section{Introduction}

DNA-based data storage has emerged as a promising technology for archival
information storage in synthetic DNA molecules, and has motivated a growing
body of work at the interface of information theory, coding, and molecular
communications \cite{Sabary2024,Milenkovic2024}. A central component of any
DNA storage system is the sequencing device used during data retrieval. Among
current sequencing technologies, the nanopore sequencer developed by Oxford
Nanopore Technologies (ONT) is particularly attractive since it combines
real-time readout, ultra-long reads, and portable low-cost hardware
\cite{ONT}.

{In nanopore sequencing, a DNA strand translocates through a
nanopore while the ionic current is sampled. The current depends on the group
of bases residing inside the pore, which we refer to as the \emph{nanopore
state}, while its expected current is the corresponding \emph{pore level}.
Repeated samples from the same state form a signal segment.}
However, from an information-theoretic perspective, the nanopore sequencer
presents a challenging read channel \cite{BITSDNA}. Firstly,
{each current sample depends on multiple bases, such that}
the input bases experience inter-symbol interference (ISI). Secondly, the
{\em translocation speed} of the bases through the pore is variable, resulting
in duplications of the nanopore state at the sample level. Therefore, the
nanopore sequencer is naturally modelled as a {\em noisy duplication channel
with memory} \cite{McBain2022,McBain2024a}, following earlier signal-level
modelling approaches for nanopore sequencers \cite{Mao2018}.
{At the signal level, the sampled current is modelled directly,
whereas a \emph{basecaller} converts this signal into an estimated base
sequence.} After basecalling, the induced sequence-level errors may
alternatively be approximated by an insertion-deletion-substitution (IDS)
channel with memory \cite{Hamoum2023,Welter2026}. Here, we study the underlying
signal-level channel. Understanding achievable information rates over this
channel is therefore a key step toward characterising the limits of
nanopore-based DNA storage systems.

This paper extends the theoretical foundations of a class of noisy duplication
channels that includes the nanopore sequencer as an important special case,
namely the noisy nanopore (duplication) channel (NNC)
\cite{McBain2022,McBain2024a}. In this model, an input state sequence evolves
according to a Markov source on a de Bruijn state-space, representing the
action of shifting bases into the nanopore, where the bases inside the nanopore
uniquely specify the channel state. Due to the random translocation speed, each
channel state is duplicated for a randomly distributed number of samples. The
duplicated channel states are mapped to reference current levels via a {\em pore
model}, and are then corrupted by additive white Gaussian noise (AWGN). The
channel output is therefore a noisy piecewise-constant signal that is stretched
by random sample duplications. Additive noise blurs the {\em jumps} between
successive levels, while the random duplications make the level-segment
boundaries unknown. Consequently, there are two principal challenges in
analysing achievable information rates for such channels: channel memory and
random sample duplications. Removing either distortion leads to either an ISI
channel or a duplication channel, both of which are already challenging in
their own right. The results developed in this paper therefore focus on
characterising how these two distortions interact in the information rate.

{
Capacity lower bounds for i.i.d.\ deletion and duplication
channels without additive noise, including geometric
duplication, were derived in \cite{DrineaMitzenmacher2007}.
For noisy nanopore channels, coding theorems established
that} the Markov-constrained capacity is the maximum
Shannon mutual information rate over all ergodic Markov
sources \cite{McBain2024b}.
This theoretical framework was
later used to empirically investigate information rates for the nanopore
simulator Scrappie \cite{McBain2024a}, and
{to evaluate achievable rates on ONT sequencing
data under mismatched decoding}
\cite{McBain2025JSAIT}. Capacity bounds for more general noisy nanopore
channels were subsequently developed in \cite{Rameshwar2025}; these bounds are
particularly useful in high-sampling-rate or low-noise regimes. Nevertheless,
the theoretical analysis of information rates remains difficult, particularly
when seeking accurate performance evaluations for nanopore systems.

The main contribution of this paper is a new decomposition of the information
rate for noisy duplication channels with geometric duplication. The
decomposition separates the effect of channel memory from the synchronisation
uncertainty caused by random sample duplications. Specifically, we show that
the information rate can be expressed as the difference between two
interpretable terms. {
While the decomposition follows directly from the chain rule, its key consequence is that the two resulting rate terms admit separate tractable
representations.} The first term is the achievable information rate of an auxiliary ISI channel,
which captures the memory of the nanopore signal at the sample level, while
the effect of sample duplications is accounted for separately. The second term
is a {\em soft alignment penalty} that captures uncertainty in the unknown boundaries
between segments, where each segment is a consecutive run of duplicated
observations generated from the same underlying current level and channel
state. This reformulation
replaces the direct analysis of the full noisy duplication channel with two more interpretable components: a hidden-Markov entropy rate term
\cite{CoverThomas} and a Soft-DTW-type alignment functional
\cite{Cuturi2017}. {For a nanopore state $s$, let $f(s)$ denote its associated
pore-model current level. We define the \emph{jump distance} between adjacent
states $s'$ and $s$ as
$J(s',s)=|f(s)-f(s')|$, which quantifies the separation between successive
nanopore levels.}

The specific contributions of this paper are as follows:
\begin{itemize}
    \item We derive a new information rate decomposition for noisy duplication channels with geometric duplication. The decomposition separates the channel memory induced by the nanopore state sequence from the synchronisation uncertainty induced by random sample duplications.

    \item We introduce a soft alignment functional, closely related to the
Soft-DTW loss, whose normalised limit characterises the soft alignment
penalty in the decomposition. We prove {a strong analogue of
the asymptotic equipartition property (AEP)} for this functional. {Together with the strong AEP for the stationary ergodic auxiliary ISI channel \cite{CoverThomas}, this yields an alternative proof of the
Markov-constrained coding theorem.}

    \item We show that the decomposition
    {yields separate Monte Carlo evaluations of its two terms}.
    The auxiliary ISI term can be evaluated using the
    {hidden-Markov} forward algorithm, while the soft alignment
    penalty can be computed using a Soft-DTW-type dynamic program.
    {The decomposition also enables {the development of lower bounds that avoid} the global soft-alignment calculation.}

    \item We derive a {first-order, genie-aided} upper bound on the soft alignment penalty in terms of
    the jump distance distribution between adjacent nanopore levels. This gives a
    lower bound on the information rate, which we call the jump-reliability
    bound, and provides a geometric explanation of nanopore synchronisability.

    \item We apply the resulting bound to estimate achievable rates for a range
    of ONT nanopore sequencers, enabling comparisons
    that were not feasible using previous techniques.
\end{itemize}

The jump-reliability bound is the central practical consequence of the decomposition. It formalises the empirical intuition that larger and more distinguishable jumps between adjacent nanopore levels improve synchronisation, and therefore improve achievable information rates \cite{McBain2022}. In this way, the proposed framework connects information-theoretic performance directly to the geometry of the pore model and the translocation dynamics of the sequencer. The resulting lower bound is especially useful in sequencing regimes of practical interest, where direct information rate estimation is computationally demanding. This is especially important for modern pore models such as R10.4.1, whose large state space makes direct information rate estimation computationally prohibitive. The jump-reliability bound reduces this analysis to the geometry of the nanopore level mapping, through statistics of adjacent level jumps.

The remainder of the paper is organised as follows. In Section~\ref{sec:soft_alignment}, we introduce the soft alignment
functionals used throughout the paper and establish the shift-subadditivity property needed later for the strong-AEP analysis. In
Section~\ref{sec:channel_model}, we define the noisy geometric duplication channel and its sample-level Markov representation. In
Section~\ref{sec:rates}, we introduce the relevant information and soft alignment rates, and prove the strong AEP for the soft alignment rate. In Section~\ref{sec:decomposition}, we derive the information rate decomposition, which expresses the achievable information rate as an auxiliary ISI channel term minus a segmentation penalty. In Section~\ref{sec:bounds}, we use this decomposition to develop computable lower bounds, including the jump-reliability bound. Finally, in Section~\ref{sec:numerical_results}, we evaluate achievable rates for noisy nanopore channels with parameters based on ONT sequencers.

\subsection{Notation}
The set of real numbers is denoted by \(\mathbb{R}\) and the set of non-negative integers by \(\mathbb{N}\); subscripts indicate additional conditions, e.g., \(\mathbb{N}_{\geq 1}\) denotes the positive integers. An arbitrary-length sequence is denoted by \((x_m)_{m\geq 1}\), while an \(m\)-length sequence is denoted by \(x^m\). {For integers \(1\leq i\leq j\), we denote the subsequence \((x_i,x_{i+1},\ldots,x_j)\) by \(x_i^j\), so that \(x^m=x_1^m\).} Sequences may also be written in bold, e.g., \(\mathbf{x}\). For a real-valued sequence \(\mathbf{x}\) and a sequence of positive integers \(\mathbf{k}\) with the same length, the sequence \(\mathbf{x}^{\mathbf{k}}\) is obtained by repeating \(x_\ell\) for \(k_\ell\) samples and concatenating the resulting blocks over \(\ell\in\{1,2,\ldots,m\}\). Throughout, \(\log\) denotes the base-$2$ logarithm and \(\ln\) denotes the natural logarithm. The binary entropy function is \(h_b(x)=-x\log(x)-(1-x)\log(1-x)\). 

\section{Soft Alignment Functionals}\label{sec:soft_alignment}

\begin{figure*}
    \centering
    \includegraphics[width=0.8\linewidth]{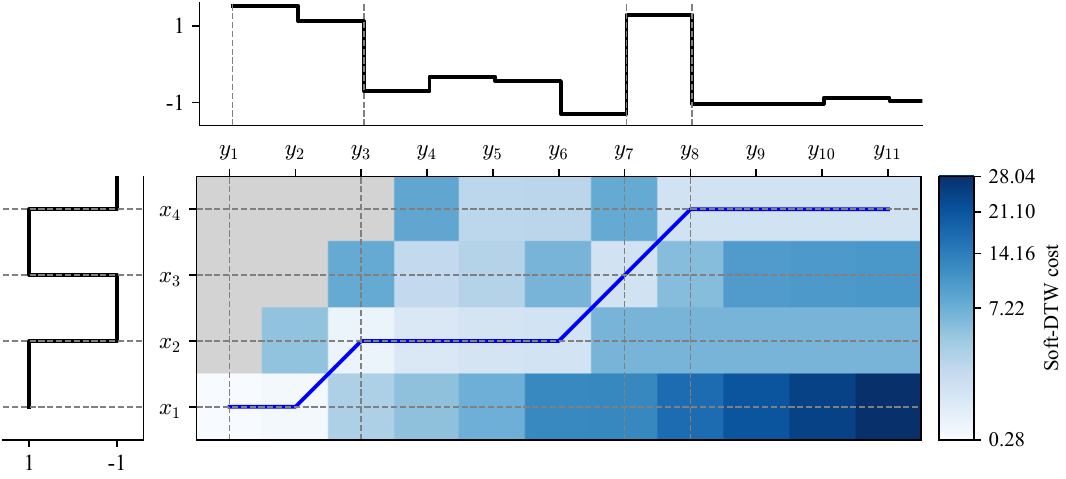}
    \caption{Example of the duplication-only Soft-DTW dynamic-programming
    matrix for piecewise-constant signals. {The signals above and to the left of the grid are the two
sequences being aligned; the levels $1$ and $-1$ are illustrative signal
levels.} Light-blue cells indicate likely alignment paths, dark-blue cells indicate unlikely paths terminating at the corresponding positions, and gray cells indicate inadmissible terminal states. The blue line shows the minimum-cost DTW path, while dashed gray lines mark segment boundaries.}
    \label{fig:softdtw}
\end{figure*}

{
The soft alignment functionals introduced in this section arise naturally
from the noisy duplication channel considered in this paper. In this channel,
a variable-length output can admit multiple possible segmentations, and the
channel likelihood requires marginalisation over these segmentations. This
marginalisation has the same soft-min structure as Soft-DTW, motivating the
alignment framework introduced here. More general soft alignment functionals
are possible, but the Soft-DTW structure is particularly natural here because
its soft-min operation matches the marginalisation over possible segmentations.
The exact correspondence with the channel likelihood is established after the
channel model is introduced
in Section \ref{sec:channel_model}.}

Soft-DTW is a smooth relaxation of dynamic time warping (DTW), which compares
two sequences while allowing for local misalignment. DTW minimises cost over
alignment paths that include local repetition, making it well suited to
settings in which repeated observations of the same state arise naturally.
In our setting, these repetitions correspond directly to sample duplications
in the channel.

The key idea of Soft-DTW is to replace the minimum over alignment costs in DTW
by a soft minimum over all admissible paths. Consequently, the resulting cost
depends not only on the single minimum-cost alignment, but on all admissible
alignments, with lower-cost paths weighted more heavily. The temperature
parameter \(\gamma\) controls this interpolation, with the hard DTW limit
recovered as \(\gamma\to0^+\).

Given sequences \(x^m\) and \(y^{t_m}\), an {\em alignment path} is a sequence
\begin{align*}
    \pi = \{(i_1,j_1), (i_2,j_2),\ldots,(i_L,j_L)\}
\end{align*}
satisfying \((i_1,j_1) = (1,1)\), \((i_L,j_L) = (m,t_m)\), and
\begin{align*}
    (i_{r+1},j_{r+1}) - (i_r,j_r) \in
        {\{(0,1),(1,1)\}}.
    \end{align*}
That is, it is a path from initial point $(1,1)$ to termination point
$(m,t_m)$ that hits $L$ coordinates in a two-dimensional lattice. The set of
all such paths is denoted by \(\Pi_{m,t_m}\). Relative to the standard
Soft-DTW formulation, we slightly restrict the admissible paths by excluding
\((1,0)\). Thus, each input symbol must be aligned to at least one output
sample, while repeated observations in the output samples are represented
through horizontal moves, and therefore $L=t_m$.
{Under this restriction, the admissible paths are in
one-to-one correspondence with admissible segmentations of the output.}

The cost of a path is
\begin{align*}
    C(\pi) = \sum_{r=1}^{L} c(x_{i_r},y_{j_r}),
\end{align*}
where \(c(\cdot,\cdot)\) is the local cost. {Throughout this
paper, we use the squared Euclidean cost
\(c(x,y)=(x-y)^2\), corresponding to a Gaussian likelihood.}

\begin{definition}[Alignment functionals \cite{Cuturi2017}] 
The Soft-DTW soft alignment functional is
\begin{align}
    \Psi_{\gamma}(x^m,y^{t_m})
    = -\gamma \ln\left(\sum_{\pi\in\Pi_{m,t_m}} e^{-C(\pi)/\gamma} \right).
\end{align}
for \(\gamma>0\). At \(\gamma = 0\), we have the DTW hard alignment functional
\begin{align}
    \Psi_{0}(x^m,y^{t_m})
    = \lim_{\gamma \rightarrow 0^+} \Psi_{\gamma}(x^m,y^{t_m}) = \min_{\pi\in \Pi_{m,t_m}} C(\pi).
\end{align}
\end{definition}

For squared Euclidean cost, the Soft-DTW functional is differentiable, although in general it is neither convex nor concave. In contrast, the low-temperature limit DTW is non-differentiable. This smoother behaviour has motivated the use of Soft-DTW in the machine learning literature~\cite{Cuturi2017}. Fig.~\ref{fig:softdtw} illustrates the computation of Soft-DTW using a dynamic-programming matrix.

The following lemma quantifies the gap between the soft alignment functional \(\Psi_\gamma\) and its hard-alignment limit \(\Psi_0\).

\begin{lemma}[Soft-min inequality \cite{Cuturi2017}]\label{lemma:soft_min_inequality}
For every \(\gamma>0\),
\begin{align}
    \Psi_0(x^m,y^{t_m}) - \gamma \ln |\Pi_{m,t_m}|
    \le
    \Psi_\gamma(x^m,y^{t_m})
    \le
    \Psi_0(x^m,y^{t_m}).
\end{align}
Hence \(\Psi_\gamma(x^m,y^{t_m}) \to \Psi_0(x^m,y^{t_m})\) as \(\gamma\to 0^+\).
\end{lemma}

A key property of alignment functionals is that they are subadditive, up to shifts in the indices of the input sequences. This leads to the following shift-subadditive lemma.

\begin{lemma}[Shift-subadditivity]\label{lemma:subadditivity}
For every \(m_1,m_2,t_1,t_2 \in \mathbb{N}\) satisfying \(t_1 \ge m_1\) and \(t_2 \ge m_2\), and every \(\gamma \in \mathbb{R}_{\ge 0}\), the alignment functional \(\Psi_\gamma\) is shift-subadditive in the sense that
\begin{align}
    \Psi_{\gamma}(x^{m_1+m_2},y^{t_1+t_2})
    &\le
    \Psi_{\gamma}(x^{m_1},y^{t_1}) +
    \Psi_{\gamma}(x_{m_1+1}^{m_1+m_2},y_{t_1+1}^{\,t_1+t_2}).
\end{align}
\end{lemma}

\begin{proof}
{Let
\(\widetilde{\Pi}^{(m_1,t_1)}_{m_1+m_2,t_1+t_2}\)
denote the subset of paths obtained by concatenating a path in
\(\Pi_{m_1,t_1}\) with a shifted path in \(\Pi_{m_2,t_2}\), connected by
the diagonal step from \((m_1,t_1)\) to \((m_1+1,t_1+1)\).}
Every such path is admissible, and its cost is the sum of the two
constituent path costs. Hence
\begin{equation*}
\begin{aligned}
\sum_{\pi\in\Pi_{m_1+m_2,t_1+t_2}}
e^{-C(\pi)/\gamma}
&\ge
\sum_{\pi\in\widetilde{\Pi}^{(m_1,t_1)}_{m_1+m_2,t_1+t_2}}
e^{-C(\pi)/\gamma}
\\
&=
\left(
\sum_{\pi_1\in\Pi_{m_1,t_1}}
e^{-C(\pi_1)/\gamma}
\right)
\times
\left(
\sum_{\pi_2\in\Pi_{m_2,t_2}}
e^{-C(\pi_2)/\gamma}
\right).
\end{aligned}
\end{equation*}
Applying \(-\gamma \ln(\cdot)\) to both sides, and using that \(-\gamma\ln(\cdot)\) is decreasing, yields the inequality.

For \(\gamma=0\), the same inequality follows by taking the limit \(\gamma \to 0^+\), or directly from the fact that restricting the minimum over \(\Pi_{m_1+m_2,t_1+t_2}\) to the subset of paths passing through \((m_1,t_1)\) can only increase its value.
\end{proof}

Although noisy duplication sequences are not ergodic in the usual sense, the shift-subadditive lemma is setting up the key structural property needed to apply Kingman’s subadditive ergodic theorem \cite{Kingman1968} when the input sequences are stochastic later in the paper.

\section{Noisy Geometric Duplication Channels}\label{sec:channel_model}

Consider an arbitrary channel input $B^m$ on alphabet $\mathcal{B}$, which uniquely specifies the channel/source state sequence $S^m$ on the state-space $\Omega$ with $|\Omega|=|\mathcal{B}|^{\tau}$ for memory constraint $\tau$ (i.e., the channel input/source depends on the current input and the previous $\tau-1$ inputs). The set of admissible state sequences is denoted by $\mathcal{S}^{(m)}$. The initial state $S_0$ has distribution $q_0(s_0)$ for all $s_0 \in \Omega$. {For the ergodic Markov sources considered below, $q_0$
denotes the stationary distribution.}

The states in a sequence $S^m$ are duplicated according to the sample durations $K^m=(K_1,\ldots,K_m)$ to give the duplicated states, or {\em sample states}, as
\begin{align}
    Z^{T_m}
    =\big(&\underbrace{S_1,\ldots,S_{1}}_{K_1},
    \underbrace{S_{2},\ldots,S_{2}}_{K_2}, \ldots, \underbrace{S_{m},\ldots,S_{m}}_{K_m}\big), \notag
\end{align}
which retains the support $\Omega$. The durations $K_{\ell}$ are i.i.d. geometric random variables ${\rm Geom}(1/\mu)$ on $\mathbb{N}_{\geq 1}$ with mean $\mu$. The number of duplication samples is $T_m =\sum_{\ell=1}^m K_{\ell}$, and the number of duplications $T_m - m$ is the negative binomial random variable ${\rm NegBin}(m,1/\mu)$ with $\mathbb{E}[T_m]=m\mu$ and $\mathrm{Var}(T_m)=m\mu(\mu-1)$. The support of $K^m$ is given by the set of admissible {\em segmentations} as $\mathcal{K}_{m,T_m} = \{k^m\in \mathbb{N}^m_{\geq 1}: \sum_{\ell=1}^{m}k_{\ell}=T_m\}$. For convenience, we define the cumulative output-index sequence
$T^m=(T_1,\ldots,T_m)$ by
\[
    T_\ell = \sum_{i=1}^{\ell} K_i, \qquad \ell=1,\ldots,m .
\]
Thus, $T_\ell$ denotes the cumulative number of output samples produced up to and including
input position $\ell$, or equivalently the index of the final output sample
generated by the first $\ell$ input symbols.

The duplicated states $Z^{T_m}{=(Z_1,\ldots,Z_{T_m})}$ are mapped to real numbers according to the {{injective}} level mapping $f:\Omega\rightarrow\mathbb{R}$, giving the level sequence $f(Z^{T_m}) = (f(Z_1),f(Z_2),\ldots,f(Z_{T_m}))$. {We denote the corresponding unduplicated level sequence by
$X^m=f(S^m)$.} This level mapping captures the ISI between bases through
the dependence of each level on the channel state.
{
Finally, let $N^{T_m}=(N_1,\ldots,N_{T_m})$, where the $N_t$ are independent
$\mathrm{Normal}(0,\sigma^2)$ random variables. The measured levels are then
}
\begin{align}
    Y^{T_m} = f(Z^{T_m}) + N^{T_m}.\notag
\end{align}
{
These simplifying assumptions exclude non-geometric or
dependent durations and state-dependent duration
distributions or noise variances. 
}
Therefore, the noisy geometric duplication channel is concisely described as in the following definition.

\begin{definition}[Noisy geometric duplication channel \cite{McBain2025JSAIT}]\label{def:nnc}
    The noisy duplication channel with AWGN and geometric duplications is  $W^{(m)} : \mathcal{S}^{(m)} \rightarrow \cup_{t\in \mathbb{N}_{\geq m}}\mathbb{R}^{t}$ with channel transition probabilities
\begin{align}\label{eq:W}
    W^{(m)}(\mathbf{y}|\mathbf{s})
    &= \alpha_m \sum_{\mathbf{k}\in\mathcal{K}_{m,t_m}} e^{-\frac{1}{2\sigma^2}||\mathbf{y} - f(\mathbf{s}^\mathbf{k})||^2}
\end{align}
with $\alpha_m = (2\pi \sigma^2)^{-\frac{t_m}{2}} (1-\frac{1}{\mu})^{t_m - m} (\frac{1}{\mu})^m$, for all inputs $\mathbf{s} \in \mathcal{S}^{(m)}$ and all channel outputs $\mathbf{y}\in \cup_{t\in \mathbb{N}_{\geq m}} \mathbb{R}^t$. 
\end{definition}

{
The sum over segmentations in \eqref{eq:W} is precisely the soft alignment
sum introduced in Section~\ref{sec:soft_alignment}. With squared Euclidean
local cost, the corresponding negative log-likelihood is therefore represented,
up to scaling and an additive constant, by the Soft-DTW functional with
temperature $\gamma=2\sigma^2$.
}

In the following, we exclusively consider time-homogeneous Markov inputs \(S^m\) on \(\Omega\) with transition matrix \(Q\), whose \((s',s)\)-entry is \(P_{S\mid S^{-}}(s\mid s')\), {where \(S^{-}\) denotes the
preceding source state}. For such sources, \((S_\ell)_{\ell\geq 1}\) is a Markov process, \((Z^{T_m})_{m\geq 1}\) is a randomly indexed semi-Markov process, and \((Y^{T_m})_{m\geq 1}\) is a randomly indexed hidden semi-Markov process~\cite{Girardin2018}. Importantly, the indexing process \((T_\ell)_{\ell\geq 1}\) is dependent on \(Z^{T_m}\), since it is the counting process corresponding to the number of samples produced after \(\ell\) input symbols.

\begin{remark}[Noisy nanopore channel]
Choosing the state space as the de~Bruijn state space \(\Omega=\{\mathsf{A},\mathsf{T},\mathsf{C},\mathsf{G}\}^\tau\), with transitions defined by shifting by one symbol and appending the next input symbol from the DNA alphabet, recovers the noisy nanopore channel in the geometric-duplication setting~\cite{McBain2024a}. More generally, allowing the duration process \((K_\ell)\) and noise process \((N_t)\) to be arbitrary yields a general noisy nanopore channel model~\cite{Rameshwar2025}.
\end{remark}

In addition, we will need the sample-level duplicated state process \((Z_t)_{t\geq 1}\), which is a semi-Markov process, and the sample-level output process \((Y_t)_{t\geq 1}\), which is a hidden semi-Markov process. Under geometric duplications, however, the sample-level duplicated state process reduces to a Markov process with transition matrix \(R\), whose \((z',z)\)-entry is \(P_{Z\mid Z^{-}}(z\mid z')\). Since
\begin{align}
P_{Z\mid Z^{-}}(z\mid z')
&= \mathbb{P}(Z_{t+1}=z\mid Z_t=z') = \left(1-\frac{1}{\mu}\right)\mathds{1}[z=z'] + \frac{1}{\mu} P_{S\mid S^{-}}(z\mid z'),\notag
\end{align}
it follows that
\begin{align}
    R = \left(1-\frac{1}{\mu}\right)I + \frac{1}{\mu}Q,
\end{align}
where \(I\) denotes the \(|\Omega|\times|\Omega|\) identity matrix.

\section{Information and Soft Alignment Rates}\label{sec:rates}

\subsection{Information Rates}
In this section, we recall the entropy and mutual information rates used to characterise noisy geometric duplication channels.
For a random process \(Y=(Y_1,Y_2,\ldots)\), the classical entropy rate is
defined as
\begin{align}
    \overline{H}(Y) = \lim_{t\to\infty} \frac{1}{t} H(Y^t),
\end{align}
when the limit exists. This definition normalises the entropy of the first
\(t\) output symbols by the number of output symbols. However, for channels
with synchronisation errors, the number of output symbols is generally not
equal to the number of input symbols. Therefore, an entropy rate definition
normalised per input symbol is required.

For the noisy duplication channel, the output produced by \(m\) channel inputs
is \(Y^{T_m}\), where \(T_m\) is the random cumulative number of output samples produced up to and including the \(m\)-th input symbol. Hence, to obtain a quantity with an operational meaning per channel input, the output entropy should be
normalised by \(m\), not by the random output length \(T_m\). Equivalently,
this can be viewed as randomly indexing the output process
\(Y=(Y_1,Y_2,\ldots)\) according to the indexing sequence
\(T=(T_1,T_2,\ldots)\), as in {\em stopped random walks} \cite{Gut2009}. We therefore define the {\em\(T\)-indexed} entropy rate of
\(Y\) as
\begin{align}
    H_T(Y) = \lim_{m\to\infty} \frac{1}{m} H(Y^{T_m}),
\end{align}
{{whose existence was established in
\cite{McBain2024b}.}}

It is important to note that \(\overline{H}(Y)\) and \(H_T(Y)\) are not
necessarily identical. The classical entropy rate \(\overline{H}(Y)\) fixes the
number of observed output symbols \(t\), but does not reveal how many input
symbols, or equivalently how many nanopore states, produced those observations.
In contrast, \(H_T(Y)\) considers the output generated by exactly \(m\) input
symbols, while the corresponding output length \(T_m\) is random. Throughout
the paper, we use this \(T\)-indexed normalisation analogously for the entropy
rates of other processes that depend on \(T\), such as \(H_T(K|Y,S)\).

When \(S^m\) is a Markov process, all \(T\)-indexed entropy rates exist by subadditivity~\cite{McBain2024b}, and they satisfy the usual basic properties of entropy rates. Furthermore, when \(Y\) is a hidden Markov or hidden semi-Markov process, as is the case for Markov sources, \(\overline{H}(Y)\) and \(H_T(Y)\) differ only by a scaling factor converting bits per sample to bits per symbol. This relationship was first established in~\cite{McBain2024b}.
Here, we give an alternative proof, adapted from~\cite{McBain2025a}, {which also establishes a strong AEP for \(H_T(Y)\) from the strong AEP for \(\overline{H}(Y)\). We specialise the argument to geometric duplications and AWGN.}

{
For notational simplicity, $p(Y^{T_m})$ denotes the density of the
random-length continuous output evaluated at its realised value.
}

\begin{theorem}[AEP for output entropy]
\label{theorem:AEP-HY}
For an ergodic Markov source, 
\begin{align}
-\frac{1}{m}\log p(Y^{T_m}) \to \mu \overline{H}(Y)
\qquad \text{a.s.}
\end{align}
as \(m\to\infty\), and
\begin{align}
H_T(Y)=\mu \overline{H}(Y).
\end{align}
\end{theorem}

\begin{proof}
{
Let $p_t(Y^t)$ denote the fixed-length marginal density of the stationary hidden-Markov
output process. The random-length output density satisfies
\begin{align}
p(Y^{T_m})
=
p_{T_m}(Y^{T_m})
\left.\mathbb{P}(T_m=t\mid Y^t)\right|_{t=T_m}.
\label{eq:random_length_density}
\end{align}
}
By the Shannon--McMillan--Breiman theorem \cite{CoverThomas}, together with
$T_m\to\infty$ and $T_m/m\to\mu$ almost surely, \cite[Theorem 8.1]{Gut2009}
\begin{align*}
-\frac{1}{m}\log p_{T_m}(Y^{T_m})=
\frac{T_m}{m}
\left(
-\frac{1}{T_m}\log p_{T_m}(Y^{T_m})
\right)
\to\mu\overline H(Y)
\qquad\text{a.s.}
\end{align*}

{
It remains to show that the endpoint probability is subexponential. For
$\epsilon>0$, let $A_t=\{\mathbb{P}(T_m=t\mid Y^t)<2^{-\epsilon m}\}.$
Then
\begin{align*}
\mathbb{P}(T_m=t,A_t)
&=\mathbb E[
\mathbf 1_{A_t}\mathbb{P}(T_m=t\mid Y^t)]
\le2^{-\epsilon m}.
\end{align*}
Thus, for any $c>\mu$,
\begin{align*}
&\mathbb{P}\!\left(
-\left.\log \mathbb{P}(T_m=t\mid Y^t)\right|_{t=T_m}>\epsilon m
\right)\le
\mathbb{P}(T_m>cm)+cm\,2^{-\epsilon m}.
\end{align*}
The right-hand side is summable for i.i.d. geometric durations using a Chernoff bound, so
Borel--Cantelli gives
\begin{align*}
-\frac{1}{m}
\left.\log \mathbb{P}(T_m=t\mid Y^t)\right|_{t=T_m}
\to0
\qquad\text{a.s.}
\end{align*}
Together with \eqref{eq:random_length_density}, this proves the strong AEP. 

Since \(p_t\) is the marginal output density, retaining the term
corresponding to the realised state path gives a lower bound on \(p_t(Y^t)\).
Together with the Gaussian upper bound $p_t(Y^t)\leq (2\pi\sigma^2)^{-t/2},$ this yields some \(C<\infty\) such that
\[
\left|-\log p_t(Y^t)\right|
\leq
C\left(1+t+\sum_{i=1}^{t}N_i^2\right).
\]
Moreover, since
\(\mathbb E[T_m^2]=
m^2\mu^2+m\mu(\mu-1)\),
\[
\mathbb E\!\left[
\left(\sum_{i=1}^{T_m}N_i^2\right)^2
\right]
=
\sigma^4\left(\mathbb E[T_m^2]+2\mathbb E[T_m]\right)
=
O(m^2).
\]
Consequently,
\[
\sup_{m\geq1}
\mathbb E\!\left[
\left|
-\frac{1}{m}\log p_{T_m}(Y^{T_m})
\right|^2
\right]
<\infty.
\]
Hence, the almost sure convergence also holds in \(L^1\):
\begin{align*}
\mathbb E\!\left[
-\frac{1}{m}\log p_{T_m}(Y^{T_m})
\right]
\to \mu\overline H(Y).
\end{align*}

For the endpoint term, concavity of $-x\log x$ gives
\begin{align*}
\mathbb E\!\left[
-\left.\log \mathbb{P}(T_m=t\mid Y^t)\right|_{t=T_m}
\right]&=
\sum_t\mathbb E\!\left[
-\mathbb{P}(T_m=t\mid Y^t)
\log \mathbb{P}(T_m=t\mid Y^t)
\right]\\
&\quad\le
\sum_t
-\mathbb{P}(T_m=t)\log \mathbb{P}(T_m=t)
=H(T_m).
\end{align*}
By Lemma~\ref{lemma:ent_obs_len}, $H(T_m)=o(m)$. Taking expectations in
\eqref{eq:random_length_density} therefore yields $H_T(Y)=\mu\overline H(Y)$.}
\end{proof}

For the sample-level Markov process \(Z\), the analogous identity \(H_T(Z)=\mu\overline{H}(Z)\) follows by the same random-length argument,
simplified by the Markov structure of \(Z\). The above theorem can therefore be viewed as an extension of this relation to the hidden Markov process \(Y\). In contrast, no corresponding extension is available after conditioning on the Markov source \(S^m\), as is required in the definition of mutual information rates. Part of the motivation for the results developed in this paper is to identify the closest analogous formulations in that setting.

The mutual information rate of the noisy geometric duplication channel is defined, whenever the limit exists, by
\begin{align}
     I_T(S;Y) = \lim_{m\to\infty} \frac{1}{m} I(S^m;Y^{T_m}).
\end{align}
Thus, \(I_T(S;Y)\) represents the asymptotic mutual information per input symbol between the input process and the random-length output process. For ergodic Markov sources, this limit exists and yields an achievable information rate~\cite{McBain2024b}. This leads to the following coding theorem under Markov input constraints.

\begin{theorem}[Markov-constrained capacity {\cite[Theorem 4]{McBain2024b}}]\label{theorem:capacity}
The capacity of the noisy geometric duplication channel constrained to ergodic Markov sources exists and is given by
\begin{align}
    C_{\rm Markov}
    = \sup_{\text{ergodic }P_{S|S^-}} I_T(S;Y).
\end{align}
\end{theorem}

Thus, \(C_{\rm Markov}\) is the supremum of achievable information rates over ergodic Markov inputs. The proof in~\cite{McBain2024b} establishes information stability using only Markov's inequality and the Shannon-McMillan-Breiman theorem, and was later presented in a slightly refined form in~\cite{McBain2025a}. In this paper, we provide an alternative proof that in part relies on the strong AEP for the output entropy rate in Theorem~\ref{theorem:AEP-HY}, together with the strong AEP-type result for soft alignment rates and the information rate decomposition developed in the subsequent sections. This upgrades the weak AEP results in~\cite{McBain2024b}, which were established with convergence in probability, to strong AEP results with almost sure convergence. 

\subsection{Soft Alignment Rates}

We now introduce soft alignment rates, which quantify the asymptotic
contribution, per input symbol, of alignment uncertainty between the input
sequence and the variable-length output. These rates are obtained by normalising the expected Soft-DTW alignment functional by the number of input symbols. They play a central role in the subsequent information rate analysis through their
relationship with $H_T(K|Y,S)$.

Moreover, $H_T(K|Y,S)$ admits a physical interpretation: it asymptotically
quantifies the logarithmic size of the set of high-probability segmentations compatible with a {\em typical} state sequence and channel output. In the
following sections, we show that the soft alignment rate appears naturally in an information rate decomposition, where it captures the penalty associated with synchronisation uncertainty and separates this effect from the ordinary
output entropy rate. It is interesting to note that, if there were only a single admissible segmentation, e.g., under the restriction $K_{\ell}=K$ for all $\ell$, the penalty would be zero.

{
The factor $m\mu\sigma^2$ in the following definition normalises the alignment cost by the expected output length and the noise variance, yielding a dimensionless quantity.
}

\begin{definition}[Soft alignment rate]
{For $\sigma>0$, the} soft alignment rate is defined as
    \begin{align}
        \overline{\psi} = \lim_{m\rightarrow\infty} \frac{1}{ m\mu \sigma^2} \mathbb{E}[\Psi_{2\sigma^2}(X^m,Y^{T_m})]~.
    \end{align}
\end{definition}

In the following lemma, we show that the soft alignment rate has a natural information-theoretic interpretation, since it is directly related to the conditional entropy rate \(H_T(K|Y,S)\).

\begin{lemma}[Soft alignment entropy relation]\label{lemma:sdtwraterelation}
The soft alignment rate $\overline{\psi}$ relates to entropy rate $H_T(K|Y,S)$ as 
\begin{align}
    \overline{\psi} = 1 - \frac{\ln(4)H_T(K|Y,S)}{\mu}~.
\end{align}
\end{lemma}
\begin{proof}
Denote by $\mathbf{S}^{\mathbf{k}}$ the duplicated state sequence $S^m$
according to segmentation $\mathbf{k}$. Observe that
$\mathbb{E}\!\left[\|Y^{T_m}-f(Z^{T_m})\|^2\right]
=\sigma^2\mathbb{E}[T_m]=\sigma^2m\mu$. Then
\begin{align*}
\frac{1}{m} H(K^m|Y^{T_m},S^m)&= -\frac{1}{m}\mathbb{E}\left[\log
\frac{e^{-\frac{1}{2\sigma^2}\|Y^{T_m}-f(Z^{T_m})\|^2}}
{\sum_{\mathbf{k}\in\mathcal{K}_{m,T_m}}
e^{-\frac{1}{2\sigma^2}\|Y^{T_m}-f(\mathbf{S}^{\mathbf{k}})\|^2}}
\right] \\
&= \frac{\mu}{\ln(4)}
-\frac{1}{m\sigma^2 \ln(4)}
\mathbb{E}\!\left[\Psi_{2\sigma^2}
\big(X^m,Y^{T_m}\big)\right]
\end{align*}
{Taking $m\to\infty$ and applying the definitions of
$H_T(K|Y,S)$ and $\overline{\psi}$ yields the result.}
\end{proof}

\begin{corollary}\label{cor:sdtw_ineq}
The soft alignment rate is bounded as
\begin{align}
{1 - \ln(4) h_b(1/\mu)} \leq \overline{\psi} \leq 1~.
\end{align}
\end{corollary}

\begin{proof}
{The bounds on soft alignment rate follow from the trivial bounds $H(K) \geq H_T(K\mid Y,S)\geq 0$.} Indeed, Lemma~\ref{lemma:sdtwraterelation} gives
\begin{align*}
H(K) \geq H_T(K\mid Y,S)
=
\frac{\mu}{\ln(4)}\left(1-\overline{\psi}\right)
\geq 0.
\end{align*}

Alternatively, the upper bound can be obtained directly from the alignment functional itself. For arbitrary inputs, the soft-min inequality in Lemma \ref{lemma:soft_min_inequality} gives
\begin{align*}
\mathrm{\Psi}_{2\sigma^2}(X^m,Y^{T_m}) \leq \mathrm{\Psi}_{0}(X^m,Y^{T_m})~.
\end{align*}
Moreover, \(\mathrm{\Psi}_{0}(X^m,Y^{T_m})\) is the minimum mean-square alignment cost over all admissible segmentations. In particular, evaluating this cost on the true segmentation \(K^m\) gives mean-square error \(\sigma^2\). Therefore
\begin{align*}
\frac{1}{m\mu}\mathbb{E}\!\left[\mathrm{\Psi}_{0}(X^m,Y^{T_m})\right]\leq \sigma^2~.
\end{align*}
By monotonicity of expectation, this proves the result.
\end{proof}

The next important result shows that soft alignment rates satisfy an analogue of the strong AEP for information rates: the normalised soft alignment functional converges almost surely to the soft alignment rate. The proof is essentially an application of Kingman's subadditive ergodic theorem \cite{Kingman1968}, which extends the classical Birkhoff ergodic theorem from additive time averages to normalised subadditive functionals, combined with properties of soft alignment functionals established in Section~\ref{sec:soft_alignment}.

\begin{theorem}[Strong AEP for the soft alignment rate]
\label{theorem:strongAEP_Psi}
For an ergodic Markov source, the normalised soft alignment functional satisfies
\begin{align}
\frac{1}{m\mu\sigma^2}\Psi_{2\sigma^2}(X^m,Y^{T_m})
\to
\overline{\psi}
\qquad \text{a.s.}
\end{align}
as \(m\to\infty\).
\end{theorem}

\begin{proof}
Define
\begin{align*}
\Lambda_m =\Psi_{2\sigma^2}(X^m,Y^{T_m}), \qquad m\in\mathbb{N}.
\end{align*}
By Lemma~\ref{lemma:subadditivity}, for every \(m_1,m_2\in\mathbb{N}\),
\begin{align*}
\Lambda_{m_1+m_2}
&=
\Psi_{2\sigma^2}(X^{m_1+m_2},Y^{T_{m_1+m_2}}) \\
&\le
\Psi_{2\sigma^2}(X^{m_1},Y^{T_{m_1}})
+
\Psi_{2\sigma^2}(X_{m_1+1}^{m_1+m_2},
Y_{T_{m_1}+1}^{T_{m_1+m_2}}).
\end{align*}

{
Let $\theta$ denote the left shift of the input-indexed marked process
$\bigl(S_\ell,K_\ell,N_{T_{\ell-1}+1}^{T_\ell}\bigr)_{\ell\geq1}$,
which is stationary ergodic. By construction,
$\Lambda_{m_2} \circ \theta^{m_1}$ equals the second term above, and hence
}
\begin{align*}
{
\Lambda_{m_1+m_2}
\leq
\Lambda_{m_1}
+
\Lambda_{m_2}\circ\theta^{m_1},
}
\end{align*}
{
which is the standard stationary subadditive form. Moreover, by
Lemma~\ref{lemma:soft_min_inequality} and evaluation of the hard
alignment cost on the true segmentation,
}
\begin{align*}
{
-2\sigma^2\ln|\Pi_{m,T_m}|
\leq \Lambda_m
\leq \|N^{T_m}\|^2 .
}
\end{align*}
{
Since $|\Pi_{m,T_m}|=\binom{T_m-1}{m-1}\leq2^{T_m}$ and
$\mathbb{E}[T_m]=m\mu$, these bounds imply
$\mathbb{E}[|\Lambda_m|]<\infty$ and a linear lower bound on
$\mathbb{E}[\Lambda_m]$.
}

Therefore, by Kingman's subadditive ergodic theorem \cite{Kingman1968},
\begin{align*}
\frac{\Lambda_m}{m}\to c \quad \text{a.s.},
\qquad
\frac{1}{m}\mathbb{E}[\Lambda_m]\to c,
\end{align*}
where \(c\) is deterministic. By the definition of \(\overline{\psi}\), this proves the result.
\end{proof}

Theorem~\ref{theorem:strongAEP_Psi} also yields a strong AEP for
\(H_T(K\mid Y,S)\). Combined with the decomposition developed in the
following section, which expresses \(I_T(S;Y)\) in terms of
\(\overline{H}(Y)\) and \(\overline{\psi}\), this completes the
alternative proof of Theorem~\ref{theorem:capacity} {
by establishing almost sure convergence of each term in the
corresponding pointwise decomposition of the information density}.

We briefly note that Theorem~\ref{theorem:strongAEP_Psi} is analogous to limit results for random path-sum models, such as {\em directed polymers in random environments} from the statistical physics literature on disordered systems \cite{Zygouras2024}. This connection suggests that tools from that literature may be useful for analysing soft alignment rates beyond the geometric setting
considered here.

\section{Information Rate Decomposition}\label{sec:decomposition}

In this section, we introduce the decomposition theorem for noisy geometric duplication channels. This theorem shows that $I_T(S;Y)$ can be written in terms of an ISI channel, with states evolving according to the sample-level Markov process, and a penalty term proportional to the soft alignment rate.
\begin{theorem} \label{theorem:decomp}
The information rate has decomposition
    \begin{align}
        I_T(S;Y) = \mu (I_{\rm ISI} - R_{\rm seg}), \label{eq:I_T_S_Y}
    \end{align}
    where 
\begin{align}
    I_{\rm ISI} &= \overline{H}(Y) - \frac{1}{2}\log(2\pi e \sigma^2),\\
    R_{\rm seg} &= h_b\left(\frac{1}{\mu}\right) - \frac{1}{\ln(4)}\left(1- \overline{\psi}\right) .
\end{align}
\end{theorem}

This decomposition has a particularly appealing computational structure,
since its two terms can be evaluated separately using standard dynamic
programming. The ISI term \(I_{\rm ISI}\) can be evaluated using the
standard forward algorithm \cite{Pfister2001,Arnold2001}.
For the sparse de~Bruijn structure, this requires
\(O(T_m|\Omega|)\) time and \(O(|\Omega|)\) memory.
The segmentation penalty \(R_{\rm seg}\) can be evaluated using Soft-DTW on
an \(m\times T_m\) grid, requiring
\(O(mT_m)\) time and \(O(m)\) memory.
{
The strong AEP results established earlier justify Monte Carlo
estimation of both terms from sufficiently long channel realisations.
}
Since \(\mathbb{E}[T_m]=m\mu\),
{the expected time complexity of Soft-DTW}
is \(O(\mu m^2)\), motivating the local jump-reliability bound
in Section~\ref{sec:bounds}.

We now present the proof of Theorem~\ref{theorem:decomp}. The key idea behind this proof is applying the chain rule to an auxiliary channel $(S^m,K^m)\mapsto Y^{T_m}$ with respect to the inputs, allowing the separation of the state information from the segmentation information. Subsequently, we show how this decomposition leads to simple yet effective lower bounds. 

\subsection{Proof {of Theorem~\ref{theorem:decomp}}}
Since $Z^{T_m}$ is a deterministic function of $(S^m,K^m)$, and the channel
output depends on $(S^m,K^m)$ only through $Z^{T_m}$, the channel is represented
by the Markov chain $(S^m,K^m)\rightarrow Z^{T_m}\rightarrow Y^{T_m}$. Hence, by conditional independence,
\[
    I(S^m,K^m;Y^{T_m}\mid Z^{T_m})=0.
\]
Moreover, since $Z^{T_m}$ is a deterministic function of $(S^m,K^m)$,
\[
    I(S^m,K^m;Y^{T_m})
    =
    I(S^m,K^m,Z^{T_m};Y^{T_m}).
\]
Therefore, by the chain rule with respect to $Z^{T_m}$,
\begin{align*}
    I(S^m,K^m;Y^{T_m})
    &= I(Z^{T_m};Y^{T_m}) + I(S^m,K^m;Y^{T_m}\mid Z^{T_m}) \\
    &= I(Z^{T_m};Y^{T_m}).
\end{align*}

Now, applying the chain rule with respect to $(S^m,K^m)$ gives
\begin{align*}
    I(S^m,K^m;Y^{T_m})
    &= I(S^m;Y^{T_m}) + I(K^m;Y^{T_m}\mid S^m).
\end{align*}
Rearranging and using the identity above, we obtain
\begin{align*}
    I(S^m;Y^{T_m})
    &= I(S^m,K^m;Y^{T_m}) - I(K^m;Y^{T_m}\mid S^m) \\
    &= I(Z^{T_m};Y^{T_m}) - I(K^m;Y^{T_m}\mid S^m).
\end{align*}

Observe that $I(Z^{T_m};Y^{T_m})=H(Y^{T_m}) - H(Y^{T_m}|Z^{T_m})$. In Theorem \ref{theorem:AEP-HY}, it was shown that $H_T(Y) = \mu \overline{H}(Y)$.
{In addition, conditioned on \(Z^{T_m}\), the output consists of \(T_m\)
independent Gaussian observations. Hence, by Wald's identity
\cite[Theorem~5.3]{Gut2009},}
\begin{align*}
    H_T(Y|Z)&=\lim_{m\rightarrow\infty}\frac{1}{m}H(Y^{T_m}|Z^{T_m})\notag\\
    &= \left(\lim_{m\rightarrow\infty}\frac{\mathbb{E}[T_m]}{m}\right) H(Y|Z)= \frac{\mu}{2}\log(2\pi e \sigma^2).
\end{align*}
Hence, the first term in the decomposition gives
\begin{align*}
    I_T(Z;Y) = \mu I_{\rm ISI},
\end{align*}
and the second term in the decomposition gives
\begin{align}\label{eq:I_KY_S}
    I_T(K;Y|S) &= H(K) - H_T(K|Y,S).
\end{align}
The first term in (\ref{eq:I_KY_S}) is the entropy of a geometric random variable, given by
\begin{align*}
    H(K) = \mu h_b\left(\frac{1}{\mu}\right),
\end{align*}
and the second term in (\ref{eq:I_KY_S}) can be represented in terms of $\overline{\psi}$ using Lemma \ref{lemma:sdtwraterelation} to give the final rate decomposition.

\section{Information Rate Bounds via Decomposition}\label{sec:bounds}

\subsection{Elementary Bounds}

The rate decomposition immediately yields simple bounds on $I_T(S;Y)$ by using only coarse bounds on the residual segmentation term. In particular, from Corollary~\ref{cor:sdtw_ineq}, or equivalently from the entropy bounds
\begin{align*}
    0 \le I_T(K;Y|S) \le H(K),
\end{align*}
we obtain
\begin{align}\label{eq:elementary_lb}
    \mu I_{\rm ISI}-H(K) \le I_T(S;Y) \le \mu I_{\rm ISI}.
\end{align}
These bounds are useful as a first benchmark, since they isolate the contribution of the associated ISI channel and require no detailed analysis of the segmentation term. In particular, removing the segmentation penalty altogether yields the simple upper bound $\mu I_{\rm ISI}$, whereas replacing the residual uncertainty by its maximum possible value $H(K)$ yields the corresponding lower bound.

However, the lower bound in \eqref{eq:elementary_lb} can be quite loose.
The reason is that even in the absence of additive noise, the channel
output need not reveal every level transition.
{
This occurs whenever two consecutive states produce the same observed
level, i.e., when $f(S_\ell)=f(S_{\ell+1})$. Since $f$ is injective,
this occurs precisely when $S_\ell=S_{\ell+1}$, corresponding to
a self-loop in the source state-space.
}
In such cases, the noiseless output still contains an intrinsic
ambiguity, and hence the source entropy cannot in general be achieved
even when $\sigma=0$.

This observation suggests using the segmentation term evaluated in the noiseless setting to sharpen the elementary lower bound uniformly over all noise levels. This motivates the following result. 

\begin{theorem}[Uniform zero-noise bound]
\label{thm:uniform_zero_noise_lower_bound}
Denote by $R_{\rm seg}(0)$ the segmentation rate $R_{\rm seg}$ term at $\sigma=0$. The information rate is bounded below as
\begin{align}
    I_T(S;Y) \ge \mu(I_{\rm ISI}-R_{\rm seg}(0)),
\end{align}
and
\begin{align}
    R_{\rm seg}(0)
    = H(Z_2|Z_1)-\frac{1}{\mu}I(S;Z).
    \label{eq:Rseg0}
\end{align}
\end{theorem}

\begin{proof}
For the lower bound on $I_T(S;Y)$ in (\ref{eq:I_T_S_Y}), replace the segmentation term in the rate decomposition by its noiseless value $R_{\rm seg}(0)$, which gives a valid lower bound uniformly over the noise level.

Regarding $R_{\rm seg}(0)$, observe that when $\sigma=0$ we have $Y=f(Z)$, and therefore $I(S;Y)=I(S;Z)$. In addition, in the noiseless setting, the associated ISI channel information rate reduces to $I_{\rm ISI}=H(Z_2|Z_1)$. Applying the rate decomposition in the noiseless case yields
\begin{align*}
    I(S;Z)=\mu\bigl(H(Z_2|Z_1)-R_{\rm seg}(0)\bigr),
\end{align*}
and rearranging gives $R_{\rm seg}(0)$ in (\ref{eq:Rseg0}).
\end{proof}

Theorem~\ref{thm:uniform_zero_noise_lower_bound} is useful because it replaces the crude entropy penalty $H(K)$ in \eqref{eq:elementary_lb} by the sharper quantity $R_{\rm seg}(0)$, which captures only the ambiguity that remains even in the absence of observation noise. Hence, it accounts for the structural non-identifiability induced by duplications of repeated identical output levels, rather than coarsely penalising by the duration entropy $H(K)$. In this sense, it gives a substantially more informative baseline lower bound than \eqref{eq:elementary_lb}. Tighter bounds for the case of $f(S_{\ell}) \neq f(S_{\ell+1})$ will be developed in the next section.

Another advantage of this bound is that $R_{\rm seg}(0)$ is computable. Indeed, the theorem shows that it suffices to compute $I(S;Z)$ in the noiseless model. This can be done using techniques recently developed in \cite{Rameshwar2025}. By contrast, extending this refinement directly to the noisy case would require computation of entropy rates for hidden Markov processes, which is generally intractable and is a classical obstacle in the computation of information rates for channels with memory \cite{Pfister2001,Arnold2001}. Thus, although the bound is based on the zero-noise model, it remains practically useful because it is both computable and uniform over all noise levels.

\begin{remark}For the special case of a binary input alphabet with distinguishable levels $\pm1$, the noiseless information rate $I(S;Z)$ reduces to the information rate of the binary sticky channel with geometric duplications \cite{Mitzenmacher2007}, which we denote by $I_{\rm sticky}$. This quantity was directly used as a bound on $I_T(S;Y)$ in \cite{McBain2024b}. For i.i.d.\ equiprobable binary inputs, Theorem~\ref{thm:uniform_zero_noise_lower_bound} gives
\begin{align*}
    R_{\rm seg}(0)
    = h_b\!\left(\frac{1}{2\mu}\right)-\frac{1}{\mu}I_{\rm sticky}.
\end{align*}
\end{remark}

\subsection{Jump-reliability Bound}

{
For comparison, the general lower bound in \cite{Rameshwar2025} incurs the
full duration-entropy penalty $H(K)$ and an additional noise penalty. Since
its noise-dependent term lower bounds $I_T(S;Y|K)\le I_T(S,K;Y)=\mu
I_{\rm ISI}$, it is no tighter than \eqref{eq:elementary_lb}. For typical ONT
parameters, $H(K)$ often exceeds the maximum source entropy of $2$ bits/base,
making this bound vacuous. The zero-noise bound improves on
\eqref{eq:elementary_lb} by replacing $H(K)$ with $\mu R_{\rm seg}(0)$. We
next develop a tighter bound that accounts for the reliability of individual
noisy level jumps.
}

The following bound is expressed in terms of the \emph{jump distance}, a quantity introduced in \cite{McBain2022,McBain2024a} as a basic measure of synchronisability. For two adjacent states $s',s\in\Omega$ with $P_{S|S^{-}}(s|s')>0$, {recall that}
\begin{align}
    J(s',s) = \bigl| f(s) - f(s') \bigr|.
\end{align}
Intuitively, larger jump distances make neighbouring segments easier to distinguish in noise, and hence reduce the uncertainty in the associated jump time. This intuition is formalised in the following theorem.

\begin{theorem}[Jump-reliability lower bound]
\label{thm:jump_reliability_lower_bound}
Let 
\begin{align}
    \overline{R}_{\rm seg}
    =
    \frac{1}{\mu}
    I(K_1;{T_2},Y^{T_2}\mid S_1,S_2).
\end{align}
Then
\begin{align}\label{eq:LB_eq}
    I_T(S;Y)\ge \mu\bigl(I_{\rm ISI}-\overline{R}_{\rm seg}\bigr).
\end{align}
Moreover,
{\begin{align}
    \overline{R}_{\rm seg}
    &=h_b\left(\frac{1}{\mu}\right)- \frac{1}{\ln 2}
    \biggl( 1-\sum_{s',s\in\Omega} q_0(s')P_{S|S^-}(s\mid s')
    \overline{\psi}(s',s)
    \biggr),
\end{align}}
{where, for $\sigma>0$,}
{\begin{align}
    \overline{\psi}(s',s)
    &=
    \frac{1}{2\mu\sigma^2}
    \mathbb{E}\!\left[
        \Psi_{2\sigma^2}\bigl((f(s'),f(s)),Y^{T_2}\bigr)
        \,\middle|\,
        S_1=s',S_2=s
    \right].
\end{align}}
\end{theorem}

\begin{proof}
Starting from the chain rule,
\begin{align*}
    I(K^m;Y^{T_m}\mid S^m)
    &= \sum_{\ell=1}^m
    I(K_\ell;Y^{T_m}\mid S^m,K^{\ell-1}) \\
    &= \sum_{\ell=1}^m
    I(K_\ell;Y_{T_{\ell-1}+1}^{T_m}
    \mid S^m,T_{\ell-1}).
\end{align*}
For each $\ell\leq m-1$, we append the neighbouring jump time
$T_{\ell+1}$ to the observations. Once $T_{\ell+1}$ is revealed, the
outputs following $T_{\ell+1}$ are conditionally independent of
$K_\ell$. Therefore,
\begin{align*}
    I(K^m;Y^{T_m}\mid S^m) &\leq
    \sum_{\ell=1}^{m-1}
    I\!\left(
        K_\ell;
        {T_{\ell+1}},Y_{T_{\ell-1}+1}^{T_{\ell+1}}
        \,\middle|\,
        S^m,T_{\ell-1}
    \right) + {H(K_m\mid S_m)} \\
    &=
    (m-1)
    I(K_1;{T_2},Y^{T_2}\mid S_1,S_2)
    +{H(K_m\mid S_m)} \\
    &=
    \mu(m-1)\overline{R}_{\rm seg}
    +{H(K_m\mid S_m)},
\end{align*}
where the second equality follows from stationarity under shifts in the
segment index and conditional independence. It follows that
\begin{align*}
    I_T(K;Y\mid S)
    &\leq
    \lim_{m\rightarrow\infty}
    \left(
        \mu\frac{m-1}{m}\overline{R}_{\rm seg}
        +{\frac{1}{m}H(K_m\mid S_m)}
    \right) =\mu\overline{R}_{\rm seg}.
\end{align*}
Substituting this upper bound on the segmentation term into the rate
decomposition yields \eqref{eq:LB_eq}.

For the expression of $\overline{R}_{\rm seg}$, we have
\begin{align*}
    \overline{R}_{\rm seg}
    &=
    \frac{1}{\mu}
    I(K_1;T_2,Y^{T_2}\mid S_1,S_2) =
    \frac{1}{\mu}
    \left[
        H(K)
        -H(K_1\mid Y^{T_2},S_1,S_2,T_2)
    \right].
\end{align*}
Conditioning further on $(S_1,S_2)=(s',s)$ gives
\begin{align}
    H(K_1\mid Y^{T_2},S_1=s',S_2=s)
    =
    \frac{\mu}{\ln 2}
    \bigl(1-\overline{\psi}(s',s)\bigr),\label{eq:two-segment-ent}
\end{align}
{
where the factor $\ln 2$ replaces the $\ln 4$ in Lemma~\ref{lemma:sdtwraterelation} since the entropy
is per two-segment block rather than per input symbol.} Averaging over $(S_1,S_2)$ completes the proof.
\end{proof}

\begin{corollary}[Dependence on jump distance]
\label{cor:psi_depends_on_jump_distance}
The quantity $\overline{\psi}(s',s)$ depends on the state pair $(s',s)$ only through the jump distance $J(s',s)$. That is, there exists a scalar function, again denoted by $\overline{\psi}$, such that
\begin{align}
    \overline{\psi}(s',s)=\overline{\psi}\bigl(J(s',s)\bigr).
\end{align}
{
Moreover, $\overline{\psi}(J)$ is nondecreasing in $J\geq0$.
}
Consequently, the bound in Theorem~\ref{thm:jump_reliability_lower_bound} depends on the channel only through the jump matrix
\begin{align}
    J=[J(s',s)]_{s',s\in\Omega}.
\end{align}
\end{corollary}
\begin{proof}
Let \(a=f(s')\), \(b=f(s)\), and \(d=b-a\). Conditioned on
\((S_1,S_2)=(s',s)\), \(Y^{T_2}\) is generated by a two-segment signal
with levels \((a,b)\). For every \(c\in\mathbb{R}\), translation invariance
of the squared Euclidean cost gives
{\(\Psi_{2\sigma^2}((a+c,b+c),Y^{T_2}+c)
=\Psi_{2\sigma^2}((a,b),Y^{T_2})\)}, so the dependence reduces to \(d\).
After taking \(c=-a\), common reflection similarly gives
{\(\Psi_{2\sigma^2}((0,-d),a-Y^{T_2})
=\Psi_{2\sigma^2}((0,d),Y^{T_2}-a)\)}.
Since the Gaussian noise is symmetric, the expected functional is the same
for \(d\) and \(-d\), and hence depends only on \(J=|d|\).

{
To show monotonicity, let \(0\leq J_1\leq J_2\) and
\(\alpha=J_1/J_2\), with \(J_2=0\) being trivial. For levels
\((0,J_2)\), define
\(\widetilde{Y}^{T_2}=\alpha Y^{T_2}+W^{T_2}\), where
\(W_i\sim\mathcal{N}(0,(1-\alpha^2)\sigma^2)\) independently.
Then \(\widetilde{Y}^{T_2}\) has the two-segment channel law with jump
\(J_1\), and
\(K_1\to(T_2,Y^{T_2})\to(T_2,\widetilde{Y}^{T_2})\).
The data-processing inequality \cite{CoverThomas}, together with
\eqref{eq:two-segment-ent}, therefore gives
\(\overline{\psi}(J_1)\leq\overline{\psi}(J_2)\).}
\end{proof}

Theorem~\ref{thm:jump_reliability_lower_bound} admits a natural genie-aided interpretation. The quantity $\overline{R}_{\rm seg}$ is a first-order upper bound on the segmentation-rate term obtained by restricting the soft alignment decoder to a two-segment observation model and revealing the neighbouring jump times as side information. More precisely, the soft alignment decoder estimates each duplication count $K_\ell$ from the two-segment observation block spanning segments $\ell$ and $\ell+1$, conditioned on the adjacent jump times that delimit this block. {Since additional side information cannot increase the residual uncertainty, $\overline{R}_{\rm seg}$ upper bounds the true segmentation penalty.}

Although $\overline{R}_{\rm seg}$ does not generally admit a closed-form
expression, it is substantially more tractable than the full segmentation-rate
term. In particular, its evaluation reduces to averaging a two-segment
quantity over state pairs $(s',s)\in\Omega^2$, rather than analysing the full
global alignment problem over an unbounded sequence of segments. Consequently,
the bound can be computed accurately by spatial averaging with relatively
modest complexity. {For each state pair, the local Soft-DTW computation involves
only two input segments and therefore has time complexity \(O(T_2)\).
Since \(\mathbb{E}[T_2]=2\mu\), its expected time complexity is
\(O(\mu)\), independent of the sequence length \(m\), in contrast to the
\(O(\mu m^2)\) expected time complexity of the global soft-alignment term.}

The main structural consequence of the theorem is given by Corollary~\ref{cor:psi_depends_on_jump_distance}, which identifies the jump distance $J(s',s)$ as the sole state-dependent parameter governing the local reliability term. Specifically, the corollary shows that $\overline{\psi}(s',s)=\overline{\psi}\bigl(J(s',s)\bigr)$, so the first-order segmentation penalty depends on the state pair $(s',s)$ only through the corresponding entry of the jump matrix. This establishes a direct relation between the achievable information rate and the geometry induced by the level mapping $f$. Table~\ref{table:psi_depends_on_jump_distance} shows \(\overline{\psi}(s',s)\) for selected jump distances and noise levels with \(\mu=10\) and \(\mu=2\){, computed via Monte Carlo integration}. From $H(K_1\mid T_2) \geq H(K_1\mid Y^{T_2},S_1,S_2,T_2) \geq 0$, and the entropy relation used in the proof of Corollary~\ref{cor:sdtw_ineq}, we obtain
\begin{align}
1-\frac{\ln(2)}{\mu}H(K_1\mid T_2)
\leq
\overline{\psi}(s',s)
\leq 1,\label{eq:twoseg-bounds}
\end{align}
where
\begin{align}
H(K_1\mid T_2)
=
\mathbb{E}\left[\log_2(T_2-1)\right].\notag
\end{align}
At high jump-to-noise ratios, \(\overline{\psi}(s',s)\) approaches
its upper bound of \(1\), whereas at low jump-to-noise ratios it
approaches the lower bounds \(0.7323\) and \(0.5534\) for
\(\mu=10\) and \(\mu=2\), respectively.
{
At $\sigma=0$, we define $\overline{\psi}(s',s)$ by taking
$\sigma\rightarrow0^{+}$ with $J(s',s)$ fixed. By
\eqref{eq:two-segment-ent}, this limit equals the upper bound
$1$ when $J(s',s)>0$ and the lower bound in \eqref{eq:twoseg-bounds} when $J(s',s)=0$.
These values give the zero-noise row of
Table~\ref{table:psi_depends_on_jump_distance}.
The positive diagonal entries correspond to $J(s',s)/\sigma=1$.
}

{
The first-order construction can be extended to higher-order bounds by
enlarging the observation block and withholding more neighbouring jump times
from the genie. Such extensions {{could potentially yield}} tighter bounds, but at the cost of
substantially increased computational complexity and reduced structural
transparency.}
By contrast, the first-order bound remains computationally feasible and depends
on a single explicitly identifiable geometric parameter, namely the jump
distance.

\begin{table}[t]
\centering
\caption{Values of $\overline{\psi}(s',s)$ for selected jump distances
$J(s',s)$ and noise levels $\sigma$. Diagonal entries are bold.}
\label{table:psi_depends_on_jump_distance}

\footnotesize
\renewcommand{\arraystretch}{1.05}

\begin{minipage}[t]{0.44\columnwidth}
\centering
\textit{(a) $\mu=10$}\\[2pt]

\begin{tabular*}{\linewidth}{
@{\extracolsep{\fill}}
c*{6}{S[table-format=1.2,detect-weight=true]}
@{}
}
\toprule
& \multicolumn{6}{c}{$J(s',s)$} \\
\cmidrule(lr){2-7}
{$\sigma$}
& {0.00} & {0.20} & {0.40} & {0.60} & {0.80} & {1.00} \\
\midrule
0.00 & {\bfseries 0.73} & 1.00 & 1.00 & 1.00 & 1.00 & 1.00 \\
0.20 & 0.73 & {\bfseries 0.82} & 0.92 & 0.97 & 0.99 & 1.00 \\
0.40 & 0.73 & 0.76 & {\bfseries 0.82} & 0.88 & 0.92 & 0.94 \\
0.60 & 0.73 & 0.75 & 0.78 & {\bfseries 0.82} & 0.86 & 0.89 \\
0.80 & 0.73 & 0.74 & 0.76 & 0.79 & {\bfseries 0.82} & 0.85 \\
1.00 & 0.73 & 0.74 & 0.75 & 0.78 & 0.80 & {\bfseries 0.82} \\
\bottomrule
\end{tabular*}
\end{minipage}%
\hspace{0.08\columnwidth}%
\begin{minipage}[t]{0.44\columnwidth}
\centering
\textit{(b) $\mu=2$}\\[2pt]

\begin{tabular*}{\linewidth}{
@{\extracolsep{\fill}}
c*{6}{S[table-format=1.2,detect-weight=true]}
@{}
}
\toprule
& \multicolumn{6}{c}{$J(s',s)$} \\
\cmidrule(lr){2-7}
{$\sigma$}
& {0.00} & {0.20} & {0.40} & {0.60} & {0.80} & {1.00} \\
\midrule
0.00 & {\bfseries 0.55} & 1.00 & 1.00 & 1.00 & 1.00 & 1.00 \\
0.20 & 0.55 & {\bfseries 0.64} & 0.80 & 0.91 & 0.97 & 0.99 \\
0.40 & 0.55 & 0.58 & {\bfseries 0.64} & 0.72 & 0.80 & 0.86 \\
0.60 & 0.55 & 0.56 & 0.60 & {\bfseries 0.64} & 0.69 & 0.75 \\
0.80 & 0.55 & 0.56 & 0.58 & 0.61 & {\bfseries 0.64} & 0.68 \\
1.00 & 0.55 & 0.56 & 0.57 & 0.59 & 0.61 & {\bfseries 0.64} \\
\bottomrule
\end{tabular*}
\end{minipage}

\end{table}

\section{Numerical Results}\label{sec:numerical_results}

\subsection{Binary-Input AWGN Channel with Geometric Duplication}\label{sec:AWGNgeom}
\begin{figure}[!b]
    \centering
    \includegraphics[width=0.8\linewidth]{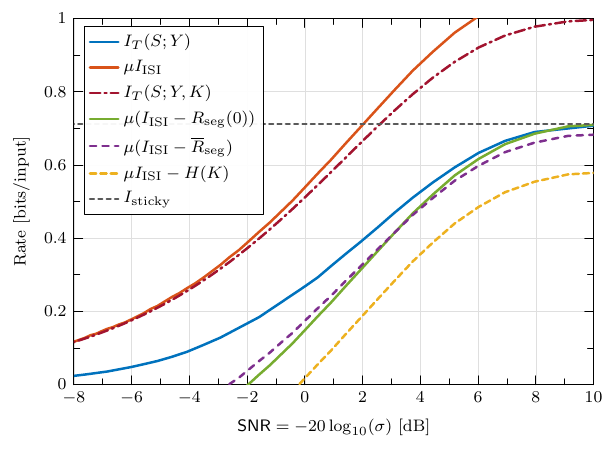}
    \caption{Bounds on achievable information rates of a noisy geometric duplication channel with $\pm 1$ binary signalling and mean duration $\mu=10/9$, driven by an i.i.d. Bernoulli-$1/2$ source. }
    \label{fig:AWGNgeom}
\end{figure}
{
Firstly, we consider a binary-input noisy geometric duplication channel for
which sticky-channel benchmarks are available in the zero-noise limit. This
provides a simple benchmark for the uniform zero-noise bound and also
illustrates a regime in which repeated identical levels make jump distance a
comparatively weak measure of synchronisability.} In particular, we consider a binary-input noisy geometric channel with state-space $\Omega=\{0,1\}$, binary signalling with $f(0)=-1$ and $f(1)=1$, and duplication probability $0.1$ ($\mu = 10/9$). The channel is driven by an i.i.d. source of Bernoulli random variables ${\rm Ber}(1/2)$. In this scenario, the sample-level Markov transition matrix is 
\[
\renewcommand{\arraystretch}{1.3}
\setlength{\arraycolsep}{8pt}
\begin{aligned}
R
&= \frac{1}{10}
\begin{bmatrix}
1 & 0 \\
0 & 1
\end{bmatrix}
+ \frac{9}{10}
\begin{bmatrix}
\frac{1}{2} & \frac{1}{2} \\
\frac{1}{2} & \frac{1}{2}
\end{bmatrix} =
\begin{bmatrix}
\frac{11}{20} & \frac{9}{20} \\
\frac{9}{20} & \frac{11}{20}
\end{bmatrix}.
\end{aligned}
\]
with stationary distribution $(q_0(0),q_0(1))=(1/2,1/2)$, which has entropy rate $H(Z_2|Z_1)=0.9928$ bits/sample. 

Information rates for this channel are shown in
Fig.~\ref{fig:AWGNgeom}.
{
The ISI and soft alignment terms are estimated from
normalised information densities and global Soft-DTW
functionals, respectively, using $m=10^6$.
The jump-reliability bound uses two-segment alignment
costs evaluated by Monte Carlo integration as in
Table~\ref{table:psi_depends_on_jump_distance}.
}
\par\newpage\noindent
{
For the specified i.i.d. Bernoulli-$1/2$ source, the noiseless
information rate $I_{\rm sticky}=0.7128$ bits/symbol (dark grey) upper bounds $I_T(S;Y)$ at every noise level by the
data-processing inequality. This is a source-specific information rate upper bound, rather than a capacity upper bound. The decomposition also gives the upper bound
$\mu I_{\rm ISI}$ (orange).} The rate $I_T(S;Y)$ (blue) equals {$I_{\rm sticky}$} when $\sigma=0$ ($\mathsf{SNR}=\infty$), and drops to almost zero at $\sigma=3$ ($\mathsf{SNR}=-9.54$ dB). Observe that the uniform zero-noise lower bound (green) is very tight at high SNRs and loosens at low SNRs. Since this lower bound is an achievable rate, it is also a lower bound on capacity, which may be tightened by optimising over Markov sources with longer memory (rather than i.i.d.). 

For reference, Fig.~\ref{fig:AWGNgeom} also shows the jump-reliability lower bound (purple). This bound is not expected to be tight for channels of this type, where repeated runs of identical levels are frequent and many adjacent jump distances are zero. In this regime, the jump distance distribution provides a weak description of synchronisability, since segment boundaries cannot be reliably inferred from level changes alone. This behaviour is reflected in the numerical results: the zero-noise bound is mildly tighter at low noise levels, whereas at higher noise levels it becomes looser than even the jump-reliability bound. Conversely, observe that the upper bound (orange) is very loose at high SNRs and becomes slightly tighter at lower SNRs. {This upper bound can be tightened at high SNRs using the genie-aided bound $I_{T}(S;Y,K)$ \cite{McBain2024a}, which is also an achievable rate of an auxiliary ISI channel.}

\subsection{Noisy Nanopore Channels with Geometric Duplication}
Next, we evaluate achievable rates for noisy nanopore channels using parameters based on real Oxford nanopore sequencers. Specifically, the level mapping $f$ is chosen from ONT-released pore models for several nanopore generations. Each mapping is normalised to have zero mean and unit variance, so that all nanopores have the same average signal power, $\mathbb{E}[X^2_1]=1${, and identical signal-to-noise ratios $\mathsf{SNR}=-20\log_{10}(\sigma)$}. The corresponding ONT sequencing devices operate at a sampling rate of $f_s=4$ kHz, with an average base translocation speed of $\nu$ bases per second (bps). Therefore, the average number of samples per base is 
\begin{align}
    \mu = \frac{f_s}{\nu} \quad \text{[samples per base]}.\notag
\end{align}
Moreover, we analyse the following Oxford nanopores: R9.2 with $\tau=6$ and $\nu=250$ bps, R9.4 with $\tau=6$ and $\nu =450$ bps, R10.4.1 with $\tau=9$ and $\nu=260$ bps, and R10.4.1 with $\tau=9$ and $\nu=400$ bps. Since the translocation speed alters the time cost, achievable rate  in terms of bits per base does not necessarily allow for a fair comparison, since a slower translocation speed means there are more samples at the sequencer output. Therefore, in the following achievable rate analysis and comparison of nanopores, we also use the {\em throughput}
\begin{align}
        \nu I\left(\frac{f_s}{\nu}\right) \quad \text{[bits per second]},\notag
\end{align}
where $I(\mu)$ is the achievable rate of a noisy nanopore channel with $\mu$ samples per base on average. {Throughput therefore reflects both the direct scaling with translocation speed
and changes in the per-base information rate caused by the different average
numbers of duplications and their resulting synchronisability.} For the following numerical results, we assume an i.i.d. uniform source. {
Using the same estimators as in Section~\ref{sec:AWGNgeom},
the ISI estimates use $m=10^5$ for R9 and
$m=10^6$ for R10.4.1, while the information-rate markers
use global Soft-DTW estimates with $m=10^6$.
}

Fig.~\ref{fig:nanopore} shows the achievable rates in bits per second for the R9 and R10 nanopores under both fast and slow translocation-speed settings, across a range of SNR values relevant to nanopore sequencing, with $\sigma \in [0.1,0.4]$.
The solid curves denote the jump-reliability lower bound in
Theorem~\ref{thm:jump_reliability_lower_bound}, evaluated over the full SNR
range, while the substantially more demanding true information rates are
evaluated only at representative points indicated by the markers. For
comparison, we also show the elementary lower bound
$\mu I_{\rm ISI}-H(K)$ and the genie-aided upper bound
$I_T(S;Y,K)$ \cite{McBain2024a}. The elementary bound is substantially
looser and becomes vacuous at lower SNRs, while the genie-aided upper bound
is relatively tight at high SNR but loosens as the noise increases. 
{
The close agreement between the markers and the solid curves shows
that the lower bound is tight at the evaluated noise levels under
the assumed channel models.
}
As the noise level increases, the bound becomes looser, suggesting that it is tight when the jump distances are sufficiently large relative to the noise level.

Using the jump-reliability lower bound, we can compare the theoretical performance of the nanopores for DNA data storage readout. As expected, faster translocation speeds lead to higher throughputs. For example, at $\mathsf{SNR}=12$ dB, the achievable rates for R9 and R10 increase by $304.24$ and $200.04$ bits/s, respectively. However, {we hypothesise that} faster translocation speeds may also increase the effective noise level by introducing greater variability in the measured levels relative to the assumed pore model, thereby reducing this gain in practice. In addition, when the rates are normalised per input base, the faster translocation-speed settings actually exhibit reductions of $0.14$ and $0.12$ bits/base for R9 and R10, respectively.

{
Moreover, Fig.~\ref{fig:nanopore} shows higher throughput
for R9.4 ($450$ bps) than R10.4.1 ($400$ bps) at their
faster translocation-speed settings.
Table~\ref{tab:jump_statistics} also shows larger median
jump distances and lower small-jump probabilities for
the R9 pore models under an i.i.d.\ uniform source and
unit-variance normalisation.
However, these comparisons do not isolate the effects
of translocation speed, state-space size, and jump
distribution on the observed throughput differences.
}

\begin{figure}
        \centering
        \includegraphics[width=0.8\textwidth]{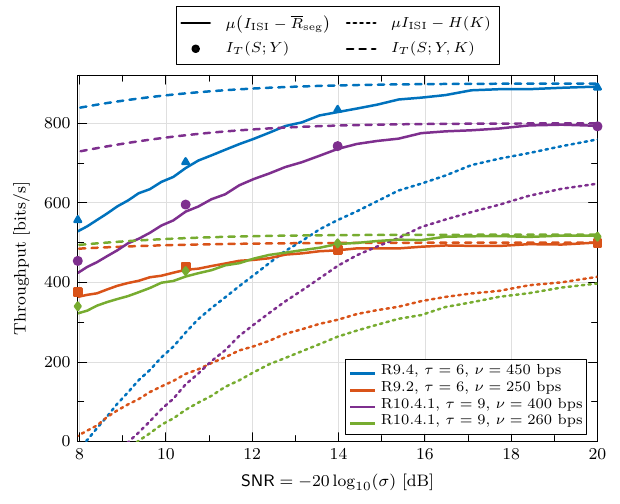}
    \caption{{Bounds on achievable information rates in bits per second
(``throughput'') for Oxford nanopore sequencers modelled as noisy nanopore
channels using ONT pore models \cite{kmermodelsONT}, with
pore levels normalised to unit variance and an i.i.d.\ uniform source.}}
    \label{fig:nanopore}
\end{figure}

\begin{table}[t]
\centering
\caption{Jump distance statistics for ONT pore models \cite{kmermodelsONT}.}
\label{tab:jump_statistics}

\footnotesize
\renewcommand{\arraystretch}{1.05}

\begin{tabular*}{0.92\columnwidth}{
@{\extracolsep{\fill}}
l c c *{4}{S[table-format=1.4]}
@{}
}
\toprule
Pore model
& $\nu$
& $\tau$
& \multicolumn{1}{c}{Mean}
& \multicolumn{1}{c}{Median}
& \multicolumn{1}{c}{$\mathbb{P}(J(S_1,S_2)<0.1)$}
& \multicolumn{1}{c}{$\mathbb{P}(J(S_1,S_2)<0.5)$} \\
\midrule
R9.4    & $450$ bps & 6 & 1.0796 & 0.9859 & 0.0466 & 0.2429 \\
R9.2    & $250$ bps & 6 & 1.0829 & 0.9968 & 0.0488 & 0.2433 \\
R10.4.1 & $260$ bps & 9 & 0.9815 & 0.8227 & 0.0646 & 0.3153 \\
R10.4.1 & $400$ bps & 9 & 1.0089 & 0.8934 & 0.0514 & 0.2894 \\
\bottomrule
\end{tabular*}
\end{table}

\section{Conclusion}

This paper developed an information rate decomposition for noisy geometric duplication channels, motivated by nanopore DNA sequencing. The decomposition separates the achievable rate into an auxiliary ISI channel term and a soft alignment penalty, thereby isolating the effects of channel memory and synchronisation uncertainty. We showed that the soft alignment term satisfies the strong AEP and is directly related to the uncertainty in the unknown segmentation.

The main practical consequence is the jump-reliability lower bound,
which links achievable rate to the distribution of level jumps
between adjacent nanopore states {through a first-order,
genie-aided upper bound on the segmentation penalty}. This bound provides a tractable geometric proxy for information rate analysis and enables theoretical comparisons of ONT pore models that would be difficult using direct information rate computation alone. In particular, it allows modern R10.4.1 pore models, with state spaces of size $4^9$, to be compared alongside earlier R9 models. The numerical results show that the bound is tight in SNR regimes relevant to nanopore sequencing and reveal a geometric mechanism for synchronisability: larger and more frequent level jumps reduce the alignment penalty and increase achievable rate.

Overall, the proposed framework turns the basic intuition that larger nanopore level jumps improve synchronisation into an information-theoretic principle, linking the geometry of the nanopore level mapping directly to achievable rates for DNA data storage readout.

\appendices
\renewcommand{\thelemma}{A\arabic{lemma}} \setcounter{lemma}{0}

\section{Random-Length Entropy Lemma}
\label{appendix:entropy_length}

\begin{lemma}\label{lemma:ent_obs_len}
Suppose \(K_1,K_2,\ldots\) are i.i.d. geometric random variables with mean \(\mu<\infty\), and let
\begin{align}
    T_m = \sum_{\ell=1}^m K_\ell .\notag
\end{align}
Then
\begin{align}
    \lim_{m\to\infty}\frac{1}{m}H(T_m)=0.\notag
\end{align}
\end{lemma}

\begin{proof}
The random variable \(T_m\) is negative binomial with mean \(\mu m\) and variance \(m\,\mu(\mu-1)\). For any integer-valued random variable \(X\) with variance $\mathrm{Var}(X)$, we have \cite{Massey1988}
\begin{align*}
    H(X)\le \frac{1}{2}\log\!\bigl(2\pi e(\mathrm{Var}(X)+1/12)\bigr).
\end{align*}
Applying this to \(X=T_m\) gives
\begin{align*}
    H(T_m)
    \le
    \frac{1}{2}\log\!\bigl(2\pi e(m\,\mu(\mu-1)+1/12)\bigr)
    = O(\log m).
\end{align*}
Dividing by \(m\) and letting \(m\to\infty\) yields the result.
\end{proof}

\bibliographystyle{IEEEtran}
\bibliography{refs}

@ARTICLE{Sabary2024,
  author={Sabary, Omer and Kiah, Han Mao and Siegel, Paul H. and Yaakobi, Eitan},
  journal={IEEE Transactions on Molecular, Biological, and Multi-Scale Communications}, 
  title={Survey for a Decade of Coding for {DNA} Storage}, 
  year={2024},
  volume={10},
  number={2},
  pages={253-271},
  doi={10.1109/TMBMC.2024.3403488}}

@ARTICLE{Milenkovic2024,
  author={Milenkovic, Olgica and Pan, Chao},
  journal={IEEE Transactions on Communications}, 
  title={{DNA}-Based Data Storage Systems: A Review of Implementations and Code Constructions}, 
  year={2024},
  volume={72},
  number={7},
  pages={3803-3828},
  doi={10.1109/TCOMM.2024.3367748}}

@ARTICLE{BITSDNA,
  author={McBain, Brendon and Viterbo, Emanuele},
  journal={IEEE BITS the Information Theory Magazine}, 
  title={An Information-Theoretic Approach to Nanopore Sequencing for {DNA} Storage}, 
  year={2023},
  volume={3},
  number={3},
  pages={95-108},
  doi={10.1109/MBITS.2024.3355883}}

@phdthesis{McBain2025a,
    author  = {McBain, Brendon},
    title   = {Coding Synthetic {DNA} for Nanopore Sequencing},
    school  = {Monash University},
    year    = {2025},
    type    = {PhD dissertation},
    url     = {https://bridges.monash.edu/articles/thesis/Coding_Synthetic_DNA_for_Nanopore_Sequencing/28139777}
}

@ARTICLE{McBain2025JSAIT,
  author={McBain, Brendon and Viterbo, Emanuele},
  journal={IEEE Journal on Selected Areas in Information Theory}, 
  title={Achievable Rates of Nanopore-Based {DNA} Storage}, 
  year={2025},
  volume={6},
  number={},
  pages={261-269},
  doi={10.1109/JSAIT.2025.3598756}}

@misc{ONT,
  author = {{Oxford Nanopore Technologies}},
  url = {https://nanoporetech.com/}
}

@article{Mao2018,
  title={Models and Information-Theoretic Bounds for Nanopore Sequencing},
  author={Wei Mao and Suhas N. Diggavi and Sreeram Kannan},
  journal={IEEE Transactions on Information Theory},
  year={2018},
  volume={64},
  pages={3216-3236}
}

@article{McBain2022,
  title={Finite-State Semi-{M}arkov Channels for Nanopore Sequencing},
  author={Brendon McBain and Emanuele Viterbo and James Saunderson},
  journal={IEEE International Symposium on Information Theory (ISIT)},
  year={2022},
  pages={216-221}
}

@INPROCEEDINGS{McBain2026,
  author={McBain, Brendon and Viterbo, Emanuele},
  booktitle={2026 IEEE International Symposium on Information Theory (ISIT)}, 
  title={Information Rate Decomposition for Noisy Geometric Duplication Channels}, 
  year={2026},
  volume={},
  number={},
  pages={1-6},
  doi={10.1109/ISIT62367.2026.11653808}}

@INPROCEEDINGS{McBain2024b,
  author={McBain, Brendon and Saunderson, James and Viterbo, Emanuele},
  booktitle={IEEE International Symposium on Information Theory (ISIT)}, 
  title={On Noisy Duplication Channels with {M}arkov Sources}, 
  year={2024},
  volume={},
  number={},
  pages={3438-3443},
  doi={10.1109/ISIT57864.2024.10619201}}

@ARTICLE{Rameshwar2025,
  author={Rameshwar, V. Arvind and Weinberger, Nir},
  journal={IEEE Journal on Selected Areas in Information Theory}, 
  title={On Achievable Rates Over Noisy Nanopore Channels}, 
  year={2025},
  volume={6},
  number={},
  pages={270-282},
  doi={10.1109/JSAIT.2025.3598773}}

@ARTICLE{Hamoum2023,
  author={Hamoum, Belaid and Dupraz, Elsa},
  journal={IEEE Access}, 
  title={Channel Model and Decoder With Memory for {DNA} Data Storage With Nanopore Sequencing}, 
  year={2023},
  volume={11},
  number={},
  pages={52075-52087},
  doi={10.1109/ACCESS.2023.3278975}
}

@ARTICLE{Mitzenmacher2007,  author={Mitzenmacher, Michael},  journal={IEEE Transactions on Information Theory},   title={Capacity Bounds for Sticky Channels},   year={2008},  volume={54},  number={1},  pages={72-77},  doi={10.1109/TIT.2007.911291}}

@ARTICLE{McBain2024a,
  author={McBain, Brendon and Viterbo, Emanuele and Saunderson, James},
  journal={IEEE Transactions on Information Theory}, 
  title={Information Rates of the Noisy Nanopore Channel}, 
  year={2024},
  volume={70},
  number={8},
  pages={5640-5652},
  doi={10.1109/TIT.2024.3404002}}

@misc{kmermodelsONT,
  author = {{Oxford Nanopore Technologies}},
  title  = {{K-mer Models}},
  howpublished = {\url{https://github.com/nanoporetech/kmer_models}},
  note   = {Accessed: 2026-05-22}
}

@inbook{Girardin2018,
author = {Girardin, Valérie and Limnios, Nikolaos},
year = {2018},
pages = {215-252},
title = {Markov and semi-Markov processes: From random sequences to stochastic processes},
isbn = {978-3-319-97411-8},
doi = {10.1007/978-3-319-97412-5_5}
}

@article{Cuturi2017,
  author    = {Marco Cuturi and Mathieu Blondel},
  title     = {Soft-{DTW}: A Differentiable Loss Function for Time-Series},
  journal = {34th International Conference on Machine Learning (ICML)},
  year      = {2017},
  pages     = {894--903}
}

@article{Pfister2001,
  title={On the achievable information rates of finite state {ISI} channels},
  author={Henry D. Pfister and Joseph B. Soriaga and Paul H. Siegel},
  journal={GLOBECOM'01. IEEE Global Telecommunications Conference (Cat. No.01CH37270)},
  year={2001},
  volume={5},
  pages={2992-2996 vol.5}
}

@article{Arnold2001,
  title={On the information rate of binary-input channels with memory},
  author={Dieter-Michael Arnold and Hans-Andrea Loeliger},
  journal={IEEE International Conference on Communications (ICC)},
  year={2001},
  volume={9},
  pages={2692-2695 vol.9}
}

@inproceedings{Massey1988,
  title={On the entropy of integer-valued random variables},
  author={Massey, James L},
  booktitle={Proceedings of 1988 Beijing International Workshop on Information Theory},
  pages={C1.1--C1.4},
  year={1988},
  address={Beijing, China}
}

@article{Zygouras2024,
  author  = {Nikos Zygouras},
  title   = {Directed Polymers in a Random Environment: A Review of the Phase Transitions},
  journal = {Stochastic Processes and their Applications},
  volume  = {177},
  pages   = {104431},
  year    = {2024}
}

@article{Kingman1968,
author = {Kingman, J. F. C.},
title = {The Ergodic Theory of Subadditive Stochastic Processes},
journal = {Journal of the Royal Statistical Society: Series B (Methodological)},
volume = {30},
number = {3},
pages = {499-510},
doi = {https://doi.org/10.1111/j.2517-6161.1968.tb00749.x},
year = {1968}
}

@book{Gut2009,
  author    = {Gut, Allan},
  title     = {Stopped Random Walks: Limit Theorems and Applications},
  edition   = {2},
  publisher = {Springer},
  address   = {New York, NY},
  year      = {2009},
  doi       = {10.1007/978-0-387-87835-5}
}

@ARTICLE{Welter2026,
  author={Welter, Lorenz and Sokolovskii, Roman and Heinis, Thomas and Wachter-Zeh, Antonia and Rosnes, Eirik and Amat, Alexandre Graell i},
  journal={IEEE Journal on Selected Areas in Information Theory}, 
  title={An End-to-End Coding Scheme for {DNA}-Based Data Storage With Nanopore-Sequenced Reads}, 
  year={2026},
  volume={7},
  number={},
  pages={17-32},
  doi={10.1109/JSAIT.2026.3655592}}

@book{CoverThomas,
  author    = {Thomas M. Cover and Joy A. Thomas},
  title     = {Elements of Information Theory},
  edition   = {2},
  publisher = {Wiley},
  address   = {Hoboken, NJ, USA},
  year      = {2006}
}

@article{DrineaMitzenmacher2007,
  author  = {Eleni Drinea and Michael Mitzenmacher},
  title   = {Improved Lower Bounds for the Capacity of {i.i.d.} Deletion and Duplication Channels},
  journal = {IEEE Transactions on Information Theory},
  volume  = {53},
  number  = {8},
  pages   = {2693--2714},
  month   = aug,
  year    = {2007},
  doi     = {10.1109/TIT.2007.901221}
}

\end{document}